\documentclass[11pt]{article}

\usepackage[margin=1in]{geometry}
\usepackage{setspace}
\usepackage{float}
\usepackage{array, makecell} %
\usepackage{xcolor,colortbl}
\usepackage{framed,color}
\usepackage{xcolor}

\usepackage{amsthm,amsmath,amssymb}
\usepackage[bookmarks=true,hypertexnames=false,pagebackref]{hyperref}
\hypersetup{colorlinks=true, citecolor=blue, linkcolor=red,
  urlcolor=blue}
\usepackage{newpxtext}
\usepackage[libertine,vvarbb]{newtxmath}

\usepackage[T1]{fontenc} 
\usepackage{textcomp} 

\usepackage[scr=rsfso]{mathalfa}
\usepackage{nicefrac}
\usepackage{cleveref}
\usepackage{graphicx}
\usepackage{aliascnt}
\usepackage[section]{placeins}
\usepackage{tcolorbox}
\usepackage{epigraph}
\usepackage{xifthen}
\usepackage{latexsym}
\usepackage{framed}
\usepackage{fullpage}
\usepackage{bm}
\usepackage{braket}

\usepackage{thmtools,thm-restate}

\usepackage{algorithm}
\usepackage{algorithmic}

\newtheorem{theorem}{Theorem}[section]

\newtheorem{factthm}[theorem]{Fact}

\newtheorem{lemma}[theorem]{Lemma}
\newtheorem{definition}[theorem]{Definition}

\newtheorem{remark}[theorem]{Remark}
\newtheorem*{remark*}{Remark}

\newcommand{\captionsetup}[1]{  }

\newcommand{\calA}{\mathcal{A}}
\newcommand{\calB}{\mathcal{B}}

\newcommand{\sfB}{\mathsf{B}}
\newcommand{\sfE}{\mathsf{E}}

\usepackage{mleftright}

\newcommand{\As}{\mathcal{A}}
\newcommand{\Bs}{\mathcal{B}}

\newcommand{\Es}{\mathcal{E}}
\newcommand{\Ts}{\mathcal{T}}

\newcommand{\Ss}{\mathcal{S}}

\newcommand{\enc}{{\sf Enc}}
\newcommand{\dec}{{\sf Dec}}

\newcommand{\gen}{{\sf Gen}}
\newcommand{\skgen}{{\sf SKGen}}
\newcommand{\pkgen}{{\sf PKGen}}

\newcommand{\A}{\mathsf{A}}
\newcommand{\B}{\mathsf{B}}
\newcommand{\C}{\mathsf{C}}

\newcommand{\E}{\mathsf{E}}
\newcommand{\D}{\mathsf{D}}

\newcommand{\ct}{{\sf ct}}

\usepackage{mdframed}

\mdfdefinestyle{figstyle}{ 
  linecolor=black!7, 
  backgroundcolor=black!7, 
  innertopmargin=10pt, 
  innerleftmargin=25pt, 
  innerrightmargin=25pt, 
  innerbottommargin=10pt 
}

\newcommand{\sk}{{\sf sk}}

\newcommand{\pk}{{\sf pk}}

\renewcommand{\kappa}{\ell}

\newcommand{\poly}{{\sf poly}}
\newcommand{\negl}{{\sf negl}}

\DeclareMathOperator{\Tr}{Tr}

\newcommand{\supp}{\mathsf{SUPP}}

\newcommand{\view}{{\sf View}}

\newcommand{\sdotfill}{\textcolor[rgb]{0.8,0.8,0.8}{\dotfill}} 

\renewcommand{\supp}{{\sf supp}}

\title{Cryptomania v.s. Minicrypt in a Quantum World}
\date{}
\author{
Longcheng Li\thanks{University of Cambridge. Email: \texttt{lilongcheng116@gmail.com}}
\quad 
Qian Li\thanks{Shenzhen International  Center For Industrial  And  Applied  Mathematics, Shenzhen Research Institute of Big Data. Email: \texttt{liqian.ict@gmail.com}}
\quad
Xingjian Li\thanks{Tsinghua University. Email: \texttt{lxj22@mails.tsinghua.edu.cn}}
\quad
Qipeng Liu\thanks{University of California San Diego. Email: \texttt{qipengliu0@gmail.com}}
}
\begin{document}
\maketitle

\begin{abstract}
We prove that it is impossible to construct perfect-complete quantum public-key encryption (QPKE) with classical keys from quantumly secure one-way functions (OWFs) in a black-box manner, resolving a long-standing open question in quantum cryptography. 

Specifically, in the quantum random oracle model (QROM), no perfect-complete QPKE scheme with classical keys, and classical/quantum ciphertext can be secure. This improves the previous works which require either unproven conjectures or imposed restrictions on key generation algorithms.
This impossibility extends to QPKE with quantum public key in natural settings, which is tight to all known QPKE constructions with quantum public key.
\end{abstract}

\section{Introduction}
 
Quantum information and computation are topics with growing importance in cryptography. They reshape people's views on cryptography drastically, including breaking classical secure cryptosystems~\cite{shor1999polynomial}, creating primitives that are impossible for classical~\cite{wiesner1983conjugate}, and weakening assumptions~\cite{bennett2014quantum}.
However, quantum cryptography is not an all-powerful tool, as it also has its own limits. Therefore, characterizing the boundary of quantum cryptography under different assumptions has become a topic of great interest.

Boundaries between classical cryptographic primitives have already been studied extensively. 
In the seminal work by Impagliazzo and Rudich~\cite{impagliazzo1989limits}, they proposed the methodology of black-box separation. They showed that one-way functions are insufficient to build public key encryption (PKE) schemes in a black-box manner. In the famous work by Impagliazzo~\cite{impagliazzoPersonalViewAveragecase1995}, he characterized five possibilities on the hardness of NP problems and their complexity consequences. Among them, he used the word `Minicrypt' to refer to a world where one-way functions exist, and the word `Cryptomania' for a world with public key cryptography primitives. 
Thus, the separation result by Impagliazzo and Rudich can be viewed as a separation between Minicrypt and Cryptomania.


It turns out that the landscape of quantum cryptography varies depending on the definition of Minicrypt. For example, it is known that with quantum communication, many primitives in classical Cryptomania can be built from one-way functions, including key agreement~\cite{bennett2014quantum}, oblivious transfer~\cite{grilo2021oblivious,bartusek2021round}, public key encryption (with quantum public keys)~\cite{coladangelo2023quantum,MW24,KMNY24,barootiPublicKeyEncryptionQuantum2023}. 
On the other hand, none of these constructions are known in the quantum computation classical communication (QCCC) setting. 
Quantum communication has various drawbacks in practice, including the difficulties of authenticating, broadcasting, reusability, and potentially adding interactions. 
In this work, we focus on the fundamental question regarding Minicrypt and Cryptomania in a quantum world \cite{hosoyamada2024finding}: 
\begin{center}
    {\it Does there exist any separation between Minicrypt and Cryptomania in the QCCC setting?}
\end{center}

\paragraph{Previous works.}

Several works have attempted to address this question, but classical proof techniques often fail due to fundamental differences between quantum and classical algorithms/information, including challenges related to cloning, rewinding, and the unique structure of quantum queries. As a result, all previous approaches have either relied on unproven conjectures or applied only to highly restricted QPKE schemes.

In the work by Austrin et al.~\cite{austrinImpossibilityKeyAgreements2022}, they initialized the study of separations between quantum key agreements and one-way functions in the quantum random oracle model (QROM). They showed that under some \emph{conjecture named `polynomial compatibility conjecture'}, quantum key agreements with perfect correctness don't exist. Since quantum PKE implies a two-round key agreement scheme, their result also implies a separation between quantum PKE and one-way functions. The same idea was also applied in separating PKE with quantum ciphertext and one-way functions~\cite{bouaziz2023towards}, which we will discuss later.
To show a separation between key agreement and one-way function in the QROM, one needs to construct an eavesdropper that breaks the security of the key agreement by making polynomially many queries to the oracle. In the paper~\cite{austrinImpossibilityKeyAgreements2022}, the authors construct an eavesdropper who only makes classical queries to the oracle, while Alice and Bob can make quantum queries in general. Because honest parties get quantum queries and the eavesdropper only classical queries, the conjecture may be too strong.

\medskip

In another line of work, Li, Li, Li, and Liu~\cite{own} approach the problem from a different perspective. They introduce tools from quantum Markov chains~\cite{fawzi2015quantum} to construct an eavesdropper for quantum public key encryption schemes that have a \emph{classical key generation process}. On a high level, they construct an Eve who uses a polynomial number of quantum queries and creates a simulated view for Alice called $\A'$, denoting the view of Alice before she tries to decrypt a ciphertext using the random oracle. Using the quantum Markov chain, they can argue that the state over the real Alice and Bob  $\A\B$ is statistically close to that over the simulated Alice and Bob $\A'\B$. Finally, they finish running the simulated Alice with one caveat: the simulated Alice may not be consistent with the real random oracle, and they have to find an appropriate oracle to complete the execution. The work~\cite{own} solves this issue by only focusing on the key generation algorithm which makes only classical queries. In a followup paper~\cite{own:ITCS}, the authors introduce a view from boolean function analysis as an attempt to attack PKE with quantum-query key generation.
They made some conjecture about the zero point distributions of low-degree polynomials and proved that the conjecture implies a separation between perfect-complete PKE and OWF in the QCCC model. 

Both works fail to establish the separation once the key-generation algorithm is allowed to quantumly query a random oracle. More broadly, lifting separations proved in the classical-accessible-oracle model to the quantum-accessible-oracle model can be notoriously difficult. While several landmark separations are known with classically accessible oracles --- e.g., quantum money v.s. one-way functions~\cite{ananth2023plausibility} and QCMA v.s. QMA~\cite{li2023classical,liu2025qma} --- whether these separations persist when the oracle admits quantum queries still remains a major open problem in quantum complexity and quantum cryptography; either these separation results completely fail when quantum access is provided, or they have to rely on unproven conjectures.

\subsection{Our Main Result}

In this work, we refine the separation framework of~\cite{own,own:ITCS} and present a unified approach to proving lower bounds and separations for multiparty-computation related primitives in the QROM. We hope the improved framework can be used for achieving lower bounds for other MPC related protocols with prefect completeness. 

Using that, we obtain a full separation between perfectly complete quantum PKE and one-way functions, i.e. the quantum version of~\cite{impagliazzo1989limits} for public-key encryption, closing the gap in previous works.
\begin{theorem}[Informal]\label{thm:thm1}
    Perfect-complete quantum public key encryption, with classical keys and classical ciphertext, does not exist in the quantum random oracle model.
\end{theorem}

\begin{remark}
Our result removes the conjecture used in~\cite{austrinImpossibilityKeyAgreements2022,own:ITCS} and the restriction of a classical key generation algorithm in~\cite{own}, leaving perfect completeness as the only requirement. 
Perfect completeness is a natural requirement satisfied by many cryptographic schemes, both classical and quantum, including all known quantum PKE schemes even with quantum keys~\cite{coladangelo2023quantum,MW24,KMNY24,barootiPublicKeyEncryptionQuantum2023}. Therefore, focusing on the perfect-complete setting does not impose a strong restriction. We conjecture that allowing non-perfectness does not change the impossibility result; we hope that the separation framework proposed in this work is helpful for removing the last restriction on QPKE.
\end{remark}

\begin{table}[t]
  \centering
  \renewcommand{\arraystretch}{1.2} 
\setlength{\tabcolsep}{6pt} 
  \begin{tabular}{|l||c|c|c|c||}
    \hline
     & \cite{austrinImpossibilityKeyAgreements2022}  & \cite{own} & \cite{own:ITCS} & This Work (\Cref{thm:thm1}) \\
    \hline
    ${\sf Gen}$ & Q &  C & Q & Q\\
    \hline
    ${\sf Enc}$ & Q & Q & Q & Q \\
    \hline
    ${\sf Dec}$ & Q & Q & Q & Q \\
    \hline
    conjecture & \checkmark & & \checkmark & \\
    \hline
  \end{tabular}
  \captionsetup{skip=8pt}
  \caption{Comparing all impossibility results on QPKE with prefect completeness in the QROM. `Q' denotes quantum, `C' denotes classical. }
  \label{tab:all_results}
\end{table}

\paragraph{Result 2: Extending to quantum ciphertext.}

Using a quantum public key can cause challenges in public-key distribution, authentication, and reusability. \cite{bouaziz2023towards} raised the open question of whether quantum PKE from one-way functions is possible when using classical keys and a quantum ciphertext; since ciphertext does not require to be distributed and reused, it does not share any difficulties of quantum public keys. They proved that it is impossible, as long as the conjecture in~\cite{austrinImpossibilityKeyAgreements2022} was true. 
We extend \Cref{thm:thm1} to the quantum ciphertext case, improving the work~\cite{bouaziz2023towards}.
\begin{theorem}\label{thm:thm3}
Perfect-complete quantum public key encryption, with classical keys and quantum ciphertext, does not exist in the quantum random oracle model.
\end{theorem}

\paragraph{Result 3: Extending to quantum public keys.}

As discussed earlier, we believe achieving reusability and non-interactivity are the core of public key encryption; thus we focused on classical keys. Still, we examine the possibility to achieve QPKE with quantum keys from one-way functions. Particularly, we show that:
\begin{theorem}\label{thm:thm2}
    Perfect-complete quantum public key encryption, with quantum public keys, classical secret keys and classical/quantum ciphertext, does not exist in the quantum random oracle model, if the public key is a deterministic function of the secret key, independent of the oracle.
\end{theorem}
Here, a public key $ \ket{\pk}$ is said to be uniquely determined by a secret key $\sk$ if it can be generated by a procedure that depends only on $\sk$ but not on the random oracle.
Notably, even when public keys depend solely on secret keys, it remains unclear how to dequantize the quantum state into a classical string. As a result, we cannot directly reduce \Cref{thm:thm2} to \Cref{thm:thm1} and must instead adopt a different approach. 

Our impossibility result is tight for all known QPKE constructions with quantum keys~\cite{coladangelo2023quantum,MW24,KMNY24,barootiPublicKeyEncryptionQuantum2023}, as they all require $\ket{\pk}$ to depend on both $\sk$ and the random oracle.

\subsection{Other Related Works}
There have also been attempts to separate QCCC key agreement from weaker primitives in ``Nanocrypt'' (primitives that are presumably weaker than one-way functions). The separation between the worlds ``Nanocrypt'' and Cryptomania in the QCCC setting has also been studied in~\cite{ananthCryptographyCommonHaar2025}, where the authors showed that there is no black box construction from long pseudorandom function-like states to QCCC key agreement schemes. 

Our result proves a black box barrier to constructing public key encryption from one-way functions in the QCCC setting. A different line of research shows that, unlike the classical hierarchy, it is also impossible to have black-box constructions of one-way functions from QCCC key agreements. In~\cite{goldinCountCryptQuantumCryptography2024}, the authors show relative to a unitary oracle, it is possible to obtain QCCC key agreement while BQP=QCMA. Another work by Kretschmer et al.~\cite{kretschmerQuantumComputableOneWayFunctions2025} shows that relative to a classical oracle, it is possible to obtain a quantum computable trapdoor one-way function, while P=NP. As trapdoor OWFs imply public key encryption, it indicates that quantum public key encryption can exist even without one-way functions in the QCCC setting.
We believe that our work together with their works provides a more complete characterization of the relation between one-way functions and other QCCC primitives.


\subsection{Technical Overview}

\subsubsection*{The quantum Markov chain method and the limitation of prior works.}

The method was first proposed by~\cite{own} in the context of quantum cryptography. They consider the following scenario, where they construct a two-round key agreement from PKE: 
\begin{enumerate}
    \item Alice first runs the key generation algorithm (with oracle access to $H$), sends the public key $m_0 := \pk$ to Bob and keeps the secret key $\sk$ herself.
    \item Bob then encrypts a random key $k$ by running the encryption algorithm (with oracle access to $H$) $\enc(\pk,k)\to \ct$, and sends back the ciphertext $m_1 := \ct$ to Alice. 
    \item Alice runs the decryption algorithm (with oracle access to $H$) to retrieve the key $k$.
\end{enumerate}
One can, without loss of generality, assume Bob has a small query weight on every input. This step can be guaranteed by running, measuring Bob multiple times and thus removing all heavy queries. We will handle it carefully in the main body; for the purpose of the overview, we simply assume Bob has no heavy queries. 

Their idea is to create some Eve such that the conditional mutual information $I(\A: \B|\E)\leq \epsilon$ by making $\poly(1/\epsilon)$ queries, which implies the three systems form an approximate Markov chain. By the operational meaning of the quantum Markov chain shown in~\cite{fawzi2015quantum}, there exists some channel $\mathcal{T}\colon \E\to\E\otimes \A'$ that generates a copy of Alice system, while guaranteeing the joint state $\sigma_{\A'\E\B}=\mathcal{T}(\rho_{\E\B})$ is $O(\sqrt{\epsilon})$ close to the original joint state $\rho_{\A\E\B}$. Here $\rho_{\A\E\B}$ and $\sigma_{\A'\E\B}$ are states right before Alice runs the decryption algorithm. The main difficulty that \cite{own,own:ITCS} aimed to solve is to execute the rest of Alice and produce the key, due to the potential inconsistency of $\A'$ and the real random oracle $H$. Intuitively, the simulated Alice's view is unlikely to match with the real random oracle; and any random oracle that is compatible with Alice's view does not automatically work with Bob's view.

In \cite{own}, they manage to solve the problem by enforcing the key generation algorithm to only make classical queries. By measuring the register $\A'$, they can generate a classical query record $R_{\A'}$ of polynomial size, and Eve will simulate the run of $\A$ on the oracle reprogrammed by $R_{\A'}$, denoted by $H^{R_{\A'}}$. They show that the simulated Alice's view, Bob's view, all transcripts and the oracle $H^{R_{\A'}}$ will be in the support of real executions of the protocol. Thus, by the perfect correctness of the protocol, we can conclude that Eve will get the correct key by running the decryption algorithm. 

We summarize their algorithm as in~\Cref{fig:LLL24_algo}. It relies on maintaining a polynomial-sized query record $R_\A$ to ensure that reprogramming does not significantly disturb Bob’s state (since we assume that Bob has no heavy query). Therefore, extending their result to a quantum key generation process seems difficult, as query transcripts are no longer well-defined in the quantum setting.

\begin{figure}[hbt!]
    \centering
\begin{mdframed}
\textbf{Algorithm (Eve’s Attack) in~\cite{own} (Informal).}
\par\textbf{Input:} random oracle $H:[N]\to\{0,1\}$ with $N=2^n$, transcripts $(m_0,m_1)$, parameter $\epsilon$.
\par\textbf{Output:} key $k_E$ such that $\Pr[k_E=k_A]=1-O(\sqrt{\epsilon})$.
\begin{itemize}
  \item \textbf{Step 1: Sample a simulated Alice's view.}
  
  Let $\A$ denote the registers of Alice, and $\B$ denote the registers of Bob right after Bob sends out the second message $m_1$.

  Run Bob repeatedly with input $m_0$ and oracle access to $H$ polynomial many times in $n, 1/\epsilon$. At this moment, $I(\A: \B|\E) \leq \varepsilon$.

  Apply the reconstruction channel
  $\mathcal{T}:\E\!\to\!\E\otimes\A'$ from~\cite{fawzi2015quantum}
  to generate a fake copy $\A'$ of $\A$, then measure $\A'$ to obtain a simulated secret key $\sk'$ and classical transcript $R_{\A'}$.

  \item \textbf{Step 2: Reprogram the oracle $H$ and run the decryption algorithm.}

  Run Alice with $m_1$, the simulated $\sk'$ and oracle access to $H^{R_{\A'}}$. Obtain a key $k_E = k'$ and output it.
\end{itemize}
\end{mdframed}
\caption{Framework of Eve's attack algorithm in \cite{own}.}
\label{fig:LLL24_algo}
\end{figure}

In a following paper~\cite{own:ITCS}, the authors introduce a view from boolean function analysis in an attempt to attack PKE with quantum key generation. It is known that the probability of Alice outputting certain $(\sk',m_0=\pk)$ can be written as a low-degree polynomial $f(H)$, here $H$ is treated as the truth table of the random oracle. 
The key observation is that if there is some polynomial-sized partial assignment $\mu$ such that $f(H^{\mu})\neq 0$ for all $H$\footnote{In other words, the oracle $H^\mu$ can still produce $(\sk, m_0)$ with a non-zero probability. With perfect completeness, we can argue that under the oracle $H^{\mu}$, Alice can still produce the correct key.}, the partial assignment $\mu$ can replace the $R_{\A'}$ in~\cite{own} and the algorithm can reprogram the oracle on the points defined by the partial assignment $\mu$. The authors make some conjecture on the existence of a distribution on such partial assignment $\mu$. They prove based on the conjecture that, there exists a separation between PKE and OWF in the QCCC model. 

\subsubsection*{Improving the separation framework.}

As explained above, the main obstacle is to construct an oracle that is consistent with Bob’s view and with the simulated view of Alice. Our contribution is a refinement of the prior framework which eventually enables us to complete the proof.
Following~\cite{own,own:ITCS}, if we can find an oracle $H'$ that meets all consistency conditions, then Eve can recover the correct key by running the decryption algorithm (as in Step 2) --- this remains true even if the decryption algorithm has unbounded queries! Even though the attack is inefficient in this case, prior attacks did not exploit the fact that decryption makes only a polynomial number of random-oracle queries.
The framework in~\Cref{fig:LLL24_algo} looks for an $H'$ such that $(\sk',H')$ in the support of  real executions. We relax this to computational closeness: i.e., $(\sk', H')$ only needs to be computationally close to some real execution. This weaker requirement still guarantees that the query-bounded algorithm (the decryption algorithm) produces outputs that are sufficiently close.

By leveraging the idea, we show that we can find $H'$ satisfying \emph{either} $(\sk', H')$ is in the support of real executions, \emph{or} $(\sk', H')$ is computationally close to some real execution. More specifically, we show a \emph{win-win situation} that Eve can efficiently find a partial assignment $\mu$:

\begin{itemize}
    \item Case (a). \emph{Either}, the probability that Alice with oracle access to $H^{\mu}$ outputs $(\sk',m_0)$ is non-zero. 
    \item Case (b). \emph{Or}, there exists $\mu'$ (may not be efficiently computable by Eve) such that the probability that Alice with oracle access to $H^{\mu \cdot \mu'}$ outputs $(\sk',m_0)$ is non-zero and $(\sk', H^{\mu \cdot \mu'})$ and $(\sk', H^{\mu})$ are indistinguishable by the decryption algorithm.
    \item Finally, $\mu$ and $\mu'$ must both be polynomial-sized. This is to guarantee that the reprogrammed oracle $H^\mu$ or $H^{\mu \cdot \mu'}$ are consistent with Bob who makes no heavy query. 
\end{itemize}

\begin{figure}[hbt!]
    \centering
\begin{mdframed}
\textbf{Algorithm (Eve’s Attack) in this work (Informal).}
\par\textbf{Input:} random oracle $H:[N]\to\{0,1\}$ with $N=2^n$, transcripts $(m_0,m_1)$, parameter $\epsilon$.
\par\textbf{Output:} key $k_E$ such that $\Pr[k_E=k_A]=1-O(\sqrt{\epsilon})$.
\begin{itemize}
  \item \textbf{Step 1: Sample a simulated Alice's view.}

   Identical to Step 1 in~\Cref{fig:LLL24_algo}, measure the simulated Alice's view to get $\sk'$.

  \item \textbf{Step 2.0: Compute the partial assignment.}

  Find a partial assignment $\mu$ that satisfies either Case (a) or Case (b) above.
    
  \item \textbf{Step 2.1: Reprogram the oracle $H$ and run the decryption algorithm.}

  Run Alice with $m_1$, the simulated $\sk'$ and oracle access to $H^{\mu}$. Obtain a key $k_E$ and output.
\end{itemize}
\end{mdframed}
\caption{Improved framework of Eve's attack algorithm in this work.}
\label{fig:our_algo}
\end{figure}

Assuming Eve can find such a $\mu$ with polynomial number of queries, the above framework will produce a key with probability close to $1$. 

\subsubsection*{Compute the partial assignment: a win-win argument.}

We discuss our main idea: how Eve can compute such a partial assignment $\mu$ of polynomial size. Let $f(H)$ be the polynomial as described above, denoting the probability that on oracle $H$, Alice outputs $(\sk',m_0)$. Since Alice only makes a polynomial number of queries, $f$ has a low degree $d$. Given the fact that on some oracle, Alice can output $(\sk',m_0)$; $f$ must be not identically zero. Furthermore, since $f$ is completely determined by $(\sk',m_0)$, Eve knows $f$. 

Eve starts by a \textbf{maximal\footnote{Here, maximal means that for every maximum monomial $w\notin S$, the set $S\cup\{w\}$ is not pairwise disjoint.} set} $S$ of \textbf{pairwise disjoint} maximum monomials $\{w_1, w_2, \ldots, w_t\}$, where a maximum monomial is defined as a monomial of degree equal to that of $f$. 
In other words, each $w_i$ consists of $\deg(f)$ variables and every $w_i, w_j$ are mutually disjoint. 
By~\cite{alon99,mid04}, for every monomial $w_i$ and every $H$, there will be a partial assignment $\mu_i$ on variables in $w_i$, such that $f(H^{\mu_i}) \neq 0$.

If $S$ consists of a lot of such monomials (say more than some threshold $m$, which is a polynomial), we can argue that the current random oracle $H$ together with $\sk'$ will be close to $(H', \sk')$  for some real executions, in the view of the decryption algorithm. This is because the decryption algorithm will only make at most a polynomial number of queries; thus, at least for some $w_i$, the decryption algorithm will have a very small query weight on any variables associated with $w_i$. By~\cite{BBBV97}, a small query weight means a small total variation distance between these two cases with oracle access to $H$ or oracle access to $H^{\mu_i}$. Therefore, Eve does not need to reprogram $H$ at all and $H$ itself (or an empty partial assignment) will work.

If $S$ only has fewer maximum monomials, every monomial may have a big query weight; thus, \cite{BBBV97} no longer works. In this case, Eve knows the total number of variables with $S$ is $t \cdot \deg(f) \leq m \cdot \deg(f)$, which is polynomial in $n$. Eve can just fix all variables in $S$. Eve finds a partial assignment $\mu$ for all variables in $S$, such that $f(H^\mu) \not\equiv 0$, i.e., the function $f(H^\mu)$ is not always zero; this again can be guaranteed by~\cite{alon99,mid04}. Since $S$ is the maximal set, fixing all variables in $S$ will reduce the degree of $f$ by at least one! Otherwise, Eve can find another monomial of degree $\deg(f)$ with all associated variables being disjoint with $S$, contradicting with $S$ being maximal.

Therefore, in this case, Eve reduces the degree of $f$ by at least one,  and only fixes at most polynomial number of variables. Eve can then repeat the whole process at most $\deg(f)$ times, either at some point it finds a lot of disjoint maximum monomials, or it reduces the degree to zero and eventually the polynomial becomes a constant non-zero function. In either case, Eve finds a satisfying $\mu$. We give the algorithm of finding a satisfying $\mu$ in \Cref{fig:compute_mu}.

\begin{figure}[hbt!]
    \centering
\begin{mdframed}
\textbf{Compute a partial assignment $\mu$.}
\par\textbf{Input:} a non-zero polynomial $f$ of degree $d$, some threshold $m$.
\par\textbf{Output:} a partial assignment $\mu$.
\begin{itemize}
  \item Initiate $\tilde{\mu} \gets \emptyset$ and $\tilde{f} \gets f$.
  \item Repeat until $\deg(\tilde{f}) = 0$:
  \begin{enumerate}
      \item Find a maximal set $\cal S$ of maximum monomials in $\tilde{f}$. 
      \item If $|\mathcal{S}| > m$: Break.
      \item Otherwise, find an assignment $\eta$ for variables in $\cal S$, such that $\tilde{f}(H^{\eta}) \not\equiv 0$.
      \begin{itemize}
          \item Let $\tilde{\mu} \gets \tilde{\mu} \cdot \eta$ and $\tilde{f}(H) := f(H^{\tilde{\mu}})$.
      \end{itemize}
  \end{enumerate}
  \item Output $\mu := \tilde{\mu}$.
\end{itemize}
\end{mdframed}
\caption{How to compute $\mu$.}
\label{fig:compute_mu}
\end{figure}

\section{Preliminaries}


We assume familiarity with the basics of quantum computing and quantum information. For a comprehensive background, we refer the reader to \cite{qcqi}. Below, we present some backgrounds that are heavily used in this
work.

\subsection{Distance Measures}

We recall the definitions of total variation distance and trace distance.

\begin{definition}[Total Variation Distance]
    Given two probability distributions $D_X$ and $D_Y$ over a finite domain
    $\mathcal{X}$, the total variation distance between them is defined as
    \[
        TV(D_X,D_Y) = \frac{1}{2} \sum_{x \in \mathcal{X}} \left| D_X(x) - D_Y(x) \right|.
    \]
    Here, $D_X(x)$ and $D_Y(x)$ denote the probability of $x$ drawn from $D_X$ and $D_Y$ respectively.
\end{definition}

\begin{definition}[Trace Distance]
    For any two quantum states $\rho$ and $\sigma$, the trace distance is
    defined by
    \[
        TD(\rho,\sigma) = \frac{1}{2}\Tr\left[\sqrt{(\rho-\sigma)^\dagger (\rho-\sigma)}\right] = \sup_{0\leq\Lambda\leq I}\Tr\left[\Lambda(\rho-\sigma)\right].
    \]
\end{definition}

The following lemma is standard (e.g., see~\cite{own}). We include the proof for completeness.
\begin{lemma}\label{lem:support} 
    For two probability distributions $D_X$ and $D_Y$ over the same classical domain, if $TV(D_X,D_Y)\leq \epsilon$, we have that
    \begin{align*}
        \Pr_{x\leftarrow D_X}[x\notin\supp(D_Y) ]\leq 2\epsilon.
    \end{align*}
\end{lemma}

\begin{proof}
       $ \sum_{x\notin\supp(D_Y) }D_X(x)\leq\sum_{x}\big|D_X(x)-D_Y(x)\big|=2TV(D_X,D_Y)\leq2\epsilon.$
\end{proof}

\subsection{Quantum Oracle Model and Random Oracle}
In the quantum oracle model, a quantum algorithm $\As$ can make quantum queries to an oracle function $H :[2^{n}] \to \{0,1\}$ via a unitary transformation $U_H$ mapping $\ket{i,b}$ to  $\ket{i,b\oplus H(i)}$. We denote such an algorithm by $\As^H$, which can be expressed as a sequence of unitaries: $U_1$, $U_H$, $U_2$, $U_H$, $\ldots$, $U_d$, $U_H$, $U_{d+1}$. Here $U_1,\ldots,U_{d+1}$ are local unitaries acting on $\As$'s internal register. 


\begin{definition}[Query Weight] \label{def:query_weight}
    Consider a quantum algorithm $\As$ that makes $d$ queries to an oracle $H$. Denote the quantum state immediately after $t$ queries to the oracle as 
\begin{align*}
    \ket{\psi_t}=\sum_{i,w}\alpha_{i,w,t}\ket{i,w},
\end{align*}
where $w$ is the content of the workspace register. Define the query weight $q_i$ of input $i$ as 
\begin{align*}
    q_i=\sum_{t=1}^d\sum_{w}|\alpha_{i,w,t}|^2.
\end{align*}
\end{definition}

\begin{lemma}[\cite{BBBV97}]\label{lem:BBBV} Consider two oracles $H,
    \tilde{H}$, and a quantum query algorithm $\As$ which makes $d$ queries. Let
    $\ket{\psi_{d}}$ and $\ket{\phi_{d}}$ denote the final state before
    measurement when running $\As$ on $H$ and $\tilde{H}$ respectively, and
    $q_i$ denote the query weight of input $i$ when running $\As$ on $H$. Then
    we have that
    \begin{align*}
        \lVert\ket{\psi_d}-\ket{\phi_d}\rVert\leq 2\sqrt{d}\sqrt{\sum_{i\colon\Tilde{H}(i)\neq H(i)}q_i}.
    \end{align*}
\end{lemma}

We will also consider the quantum random oracle model (QROM). In this setting, a quantum algorithm has access to a random oracle $H:[2^{n_\lambda}]\to\{0,1\}$, which is chosen from the uniformly random distribution over all functions mapping $[2^{n_{\lambda}}]$ to $\{0,1\}$.

\subsection{Entropy and Information}
\begin{definition}[Von Neumann Entropy]
    Let $\rho \in \mathbb{C}^{2^n}$ be a quantum state describing the system
    $\A$, and let $\ket{\phi_1}, \ket{\phi_2}, \dots, \ket{\phi_{2^n}}$ be an
    eigenbasis for $\rho$, so that 
    \[
        \rho = \sum_i \eta_i \ket{\phi_i}\bra{\phi_i}.
    \]
    The Von Neumann entropy of $\rho$, denoted by $S(\rho)$ or $S(\A)_\rho$, is
    defined as
    \[
        S(\A)_\rho = S(\rho) = -\sum_i \eta_i \log(\eta_i).
    \]
    For a composite system $\A\B$ with joint state $\rho_{\A\B}$, the conditional Von Neumann entropy is defined by
    \[
        S(\A|\B)_\rho = S(\A\B)_\rho - S(\B)_\rho.
    \]
\end{definition}

In the following, we will omit the subscript $\rho$ when the quantum state is
clear from context. For example, we will write $S(\A)$ instead of $S(\A)_\rho$,
and $I(\A:\B)$ instead of $I(\A:\B)_\rho$.

\begin{definition}[Mutual Information]
    Given a quantum state $\rho$ that describes the joint systems $\A$ and $\B$,
    the mutual information between $\A$ and $\B$ is given by
    \[
        I(\A:\B) = S(\A) + S(\B) - S(\A\B).
    \]
\end{definition}

\begin{definition}[Conditional Mutual Information, CMI]
    Let $\rho$ be a quantum state describing the three joint systems $\A$, $\B$, and $\C$. Then the conditional mutual information is defined as
    \[
        I(\A:\B|\C) = S(\A\C) + S(\B\C) - S(\A\B\C) - S(\C).
    \]
\end{definition}

The strong subadditivity property states that both the mutual information and
the conditional mutual information are always non-negative.
\begin{lemma}[Strong Subadditivity, \cite{araki1970entropy}]
\label{lem:strong_subadditivity}
    For Hilbert spaces $\A$, $\B$, and $\C$, it holds that
    \[
        S(\A\C) + S(\B\C) \geq S(\A\B\C) + S(\C).
    \]
    In its conditional form, for Hilbert spaces $\A$, $\B$, $\C$, and $\D$, we have
    \[
        S(\A\C|\D) + S(\B\C|\D) \geq S(\A\B\C|\D) + S(\C|\D).
    \]
\end{lemma}

Fawzi and Renner \cite{fawzi2015quantum} provided an insightful characterization
of quantum states when the conditional mutual information is nearly zero.
Intuitively, a small value of $I(\A:\B|\E)$ indicates that the system $\A$ can
be approximately reconstructed from system $\E$. Formally, 
\begin{theorem}[Approximate Quantum Markov Chain, \cite{fawzi2015quantum}]\label{thm:quantum_mutual_info_operational}
    For any state $\rho_{{\sf AEB}}$ over systems $\A\E\B$, there exists a explicitly constructible channel ${\cal T} : {\sf E} \to {\sf E} \otimes {\sf A}'$ such that the
    trace distance between the reconstructed state $\sigma_{\sf A'EB} = {\cal
    T}(\rho_{\E\B})$ and the original state $\rho_{\sf AEB}$ is at most 
    \begin{align*}
        \sqrt{\ln 2 \cdot I(\A:\B|\E)_\rho}.
    \end{align*}
\end{theorem}

\subsection{Notations in Boolean Function Analysis}

Any function $f:\{0,1\}^N\rightarrow\mathbb{R}$ has a unique expression as a
multilinear polynomial
\[
f(x)=\sum_{S\subseteq[N]}a_S\cdot x_S,
\]
where $x_S:=\Pi_{i\in S} x_i$, and the coefficient $a_S$ is given by $a_S=2^{-N}\sum_{x}f(x)\cdot x_S$. The
\emph{degree} of $f$, denoted $\deg(f)$, is defined as the degree of its
multilinear polynomial expression, i.e., $\max\{|S|: a_S\neq 0\}$. A monomial
$x_S$ is called \emph{maximum} if $a_S\neq 0$ and it has degree $\deg(f)$, i.e., $|S|=\deg(f)$. Two monomials $x_S$ and $x_T$ are called \emph{disjoint} if
$S\cap T=\emptyset$. We say that $f$ is \emph{not identically zero} if $f(x)\not\equiv0$.

A \emph{partial assignment} is a function $\mu:[N]\rightarrow\{0,1,\star\}$. We
define the support of $\mu$ as $\supp(\mu):=\{i|\mu(i)\neq \star\}$, and the
size as $|\mu|:=|\supp(\mu)|$. $\mu$ is called \emph{empty} if $|\mu|=0$. 
For
$x\in\{0,1\}^N$, we define the modification of $x$ with $\mu$, denoted by
$x^\mu$, as the string $x'\in\{0,1\}^N$ such that
\[
x'_i:=\begin{cases}
    \mu(i) & \text{if } i\in\supp(\mu),\\
    x_i & \text{otherwise}.
\end{cases}
\]
Given two partial assignments $\mu$ and $\eta$, define their \emph{product}, denoted by $\mu\cdot\eta$, to be the partial assignment satisfying that
$x^{\mu\cdot\eta}=(x^{\mu})^{\eta}$ for any $x\in\{0,1\}^N$. Note that the
product operator is associative but not commutative.
We say that two partial assignments $\mu$ and $\eta$ are \emph{disjoint} if
$\supp(\mu)\cap \supp(\eta)=\emptyset$. 


\begin{lemma}[\cite{BBCMW01}]\label{lem:poly_method}
    Suppose a quantum algorithm makes $d$ queries to a Boolean string\footnote{We interpret the string $x \in \{0,1\}^N$ as a function $x:[N] \to\{0,1\}$, and model queries to $x$ as oracle queries to this function.} $x\in\{0,1\}^N$,  and the acceptance probability is denoted by $f(x)$. Then the function $f:\{0,1\}^N\rightarrow \mathbb{R}$ has degree at most $2d$. That is, $f$ can be expressed as 
    \[
    f(x)=\sum_{|S|\leq 2d} a_{S}\cdot x_S.
    \]
\end{lemma}

\begin{lemma}[\cite{mid04}] \label{lem:alon} Let
    $f:\{0,1\}^N\rightarrow\mathbb{R}$ be any function that is not identically zero, and $x_S$ be any maximum monomial of $f$. For any $x\in\{0,1\}^N$,
    there exists a $\mu$ with $\supp(\mu)=S$ such that $f(x^\mu)\neq 0$. 
\end{lemma}

\subsection{Quantum Public-Key Encryption}

This section provides the formal definition of Quantum Public-Key Encryption (QPKE) and Quantum Key Agreement (QKA) in QROM. 

\begin{definition} \label{def:PKE} Let $\lambda \in \mathbb{Z}_+$ be the
    security parameter and $H\colon [2^{n_\lambda}]\to\{0,1\}$ be a random
    oracle. A quantum public-key encryption scheme, relative to $H$, consists of
    the following three bounded-query quantum algorithms:
    \begin{itemize}
        \item $\gen^{H}(1^\lambda)\to(\pk,\sk)$: the key generation algorithm
        that generates a pair of classical public key $\pk$ and classical secret
        key $\sk$.
        \item $\enc^{H}(\pk,m)\to \ct$: the encryption algorithm that takes as input the
        public key $\pk$ and the plaintext $m$, and produces a classical ciphertext $\ct$. 
        \item $\dec^{H}(\sk,\ct)\to m'$: the decryption algorithm that takes as input
        the secret key $\sk$ and the ciphertext $\ct$, and outputs the plaintext $m'$. 
    \end{itemize}
    The algorithms need to satisfy the following requirements:
    \begin{description}
        \item[Perfect Completeness]
        $\Pr\left[\dec^{H}\left(\sk,\enc^{H}(\pk,m)\right)=m\colon
        \gen^H(1^\lambda)\to(\pk,\sk)\right]=1$.
        \item[IND-CPA Security] For any adversary $\mathcal{E}^{H}$ that
        makes $\poly(\lambda)$ quantum queries, for every two plaintexts 
        $m_0\neq m_1$ chosen by $\Es^{H}$ after seeing $\pk$, we have
        \begin{align*}
            \Pr_{b\leftarrow\{0,1\}}\left[\mathcal{E}^{H}\left(\pk,\enc^{H}(\pk,m_b)\right)=b\right]\leq \frac{1}{2}+\negl(\lambda).
        \end{align*}
    \end{description}    
\end{definition}


For simplicity, we use ``QPKE'' to refer to quantum public-key encryption
schemes with classical secret key, public key, and ciphertext, unless specified
otherwise. Besides, we will also consider QPKE schemes with quantum public keys, defined as follows. 

\begin{definition}[QPKE with quantum public key]\label{def:PKE_qpk}
Let $\lambda \in \mathbb{Z}_+$ be the security parameter and $H\colon
[2^{n_\lambda}]\to\{0,1\}$ be a random oracle. A quantum public-key encryption
scheme with quantum public key, relative to $H$, consists of the following four
bounded-query quantum algorithms:
\begin{itemize}
    \item $\skgen^{{H}}(1^\lambda)\to\sk$: the secret key generation algorithm
    that generates a classical secret key $\sk$. 
    \item $\pkgen^{{H}}(\sk)\to\rho_{\pk}$: the public key generation algorithm
    that takes the secret key $\sk$ and generates a quantum state
    $\rho_{\pk}$ as the public key.
    \item $\enc^{H}(\rho_{\pk},m)\to \ct$: the quantum encryption
    algorithm that takes the public key $\rho_{\pk}$ and the plaintext $m$, and produces a
    classical or quantum ciphertext $\rho_{\ct}$. 
    \item $\dec^{H}(\sk,\rho_{\ct})\to m'$: the quantum decryption
    algorithm that takes the secret key $\sk$ and the ciphertext $\rho_{\ct}$, and outputs the
    plaintext $m'$. 
\end{itemize}
The algorithms need to satisfy the following requirements:
\begin{description}
    \item[Perfect Completeness] $$\Pr\left[\dec^{H}\left(\sk,\enc^{H}(\rho_{\pk},m)\right)=m\colon \skgen^{{H}}(1^\lambda)\to\sk,\pkgen^H(\sk)\to\rho_{\pk}\right]=1.$$
    \item[IND-CPA Security] For any adversary $\mathcal{E}^{H}$ that receives $\poly(\lambda)$ copies of the public key and can make $\poly(\lambda)$ queries, and every two plaintexts $ m_0\neq m_1$ chosen by the adversary, we have
    \begin{align*}
        \Pr_{b\leftarrow\{0,1\}}\left[\mathcal{E}^{H}\left(\rho_{\pk}^{\otimes \poly(\lambda)},\enc^{H}(\rho_{\pk},m_b)\right)=b\right]\leq \frac{1}{2}+\negl(\lambda).
    \end{align*}
    
\end{description}
\end{definition}

In this paper, we focus on the setting where the quantum public key is uniquely determined by the secret key $\sk$; that is, the quantum algorithm $\pkgen(\sk)$ makes no queries to the oracle $H$. 
This setting covers all possible QPKE schemes with classical public keys, as we may assume, without loss of generality, that $\sk$ contains a copy of $\pk$. Furthermore, we may assume that $\rho_{\pk}$ is a pure state, since $\sk$ can be taken to include a purification of $\rho_{\pk}$\footnote{Assume $\pkgen(\sk)$ first generates $\sum_{x}\alpha_x\ket{x}\ket{\pk_x}$, and then discards $\ket{x}$ to obtain $\rho_\pk:=\sum_{x}|\alpha_x|^2\ket{\pk_x}\bra{\pk_x}$. Without affecting the scheme's functionality, we can include a copy of $x$ in $\sk$ so that the public key becomes a pure state $\ket{\pk_x}$. }.


Lastly, we define Quantum Key Agreement in QROM.
\begin{definition}[Quantum Key Agreement in the Oracle Model] 
    Let $\lambda\in\mathbb{Z}_+$ be the security parameter and let
    $H:[2^{n_\lambda}]\to\{0,1\}$ be a random oracle. A Quantum Key Agreement (QKA)
    protocol involves two parties, Alice and Bob, who initially begin with
    all-zero states. They can perform any quantum operations, make
    $\poly(\lambda)$ quantum queries to the oracle $H$, and exchange classical
    messages. At the end of the protocol, Alice and Bob output classical strings
    $k_A$ and $k_B$, respectively. 
    
    The protocol needs to satisfy the following conditions: 
    \begin{description}
        \item[Correctness] $\Pr[k_A = k_B] \geq 1/\poly(\lambda)$, where the
        probability is taken over the randomness of Alice and Bob's channels, and
        the random oracle $H$.  
        \item[Security] For any eavesdropper Eve that makes  $\poly(\lambda)$
        quantum queries to $H$, eavesdrops on classical communication between Alice and
        Bob and outputs $k_E$, we have $\Pr[k_A = k_E]=\negl(\lambda)$.
    \end{description}
\end{definition}

Similar to the perfect completeness in QPKE, a QKA protocol is said to be \emph{perfect complete} if it satisfies $\Pr[k_A=k_B]=1$.

\section{Helper Lemmas}
This section presents some helper lemmas, which may be of independent interest. 
\subsection{Information-Theoretic Tools}\label{sec:helper_lemmas1}

The following two information-theoretic lemmas from \cite{own} will be used to
prove our main results. We provide their proofs in
\Cref{sec:proofs} to make this paper self-contained. We assume the quantum algorithm always defers measurements to the end without loss of generality.

Lemma \ref{lem:entropyB} upper bounds how much entropy a quantum algorithm can accumulate through oracle queries.

\begin{lemma}[\cite{own}] \label{lem:entropyB} 
    Consider an algorithm $\calA$ that starts with a pure state, and makes $d$ quantum queries to a random
    oracle $H:[2^{n}]\to\{0,1\}$ without intermediate measurements. Let $\A$ denote the whole register of $\calA$ and let
    $\rho$ be the quantum state right before the final measurement. Then, it holds
    that $S(\A)_\rho\leq 2d(n+1)$. 
\end{lemma}

Lemma \ref{lemma:permuation-invariance} claims that repetition decreases CMI.
\begin{definition}[Permutation Invariance]
    Let $\A, \B_1,\B_2,\cdots,\B_t$ be a $(t+1)$-partite quantum system.
    Given the joint state $\rho_{\A\B_1\B_2\cdots \B_t}$, we say that $\B_1,
    \dots, \B_t$ are permutation invariant if, for any permutation $\pi$ on
    $[t]$, it holds that  
    \[
    \rho_{\A\B_1\B_2\cdots \B_t} = \rho_{\A\B_{\pi(1)}\B_{\pi(2)}\cdots \B_{\pi(t)}}.
    \]
\end{definition}

\begin{lemma}[Lemma 4.2, \cite{own}]\label{lemma:permuation-invariance} 
Let $\A,\B_1, \B_2, \dots, \B_t, \C$ be a $(t+2)$-partite quantum
system. Suppose the joint state $\rho_{\A\C\B_1\B_2\cdots \B_t}$ is fully
separable, i.e., it can be written as $\sum_{i} p_i\rho_{A}^{(i)}\otimes\rho_{C}^{(i)}\otimes\rho_{B_1}^{(i)}\otimes\cdots\otimes\rho_{B_t}^{(i)}$. If $\B_1, \B_2, \dots, \B_t$ are permutation invariant, then there
exists some $0 \leq j \leq t-1$ such that\footnote{We remark that $j$ cannot be simply set as $t-1$, since conditioning on more registers may increase CMI.}
\[
I(\B_t: \A \mid \C, \B_1, \dots, \B_j)_{\rho} \leq \frac{S(\A)}{t}.
\]
\end{lemma}

\subsection{A Structural Property of Low-Degree Polynomials}
In the proof of our main results, a key step involves explicitly reprogramming an unknown $x\in\{0,1\}^N$ (representing the oracle) by modifying at most $\poly(d)$ bits, in order to make a given degree-$d$ polynomial $f$ (representing the probability that $\gen$ or $\skgen$ outputs a particular $\sk$) evaluate to non-zero.

The following lemma, which builds heavily on \Cref{lem:alon}, establishes a win-win situation: by obliviously modifying a small number of bits of the unknown $x$, denoted by a partial assignment $\mu$, either we can already guarantee that $f(x^\mu)\neq0$, or there must exist many \emph{disjoint} (albeit unknown) partial assignments to make $f$ evaluate to non-zero.


\begin{lemma} \label{lem:reprogram}    
    Let $m>0$ be an integer. For any degree-$d$ function
    $f:\{0,1\}^N\rightarrow\mathbb{R}$ that is not identically zero, we can explicitly construct a partial assignment $\mu$
    of $|\mu|\leq md^2$ such that: either
    \begin{enumerate}
    \item[(a)] for any $x\in\{0,1\}^N$, $f(x^\mu)\not=0$; or
    \item[(b)] for any $x\in\{0,1\}^N$, there must exist $m$ pairwise disjoint partial assignments
    $\mu_1,\ldots, \mu_{m}$ of size at most $d$ such that $f(x^{\mu_\ell\cdot\mu})\neq
    0$ for all $\ell\in[m]$.
    \end{enumerate}
\end{lemma}

\begin{proof}
We propose an algorithm to construct such a partial assignment $\mu$. The algorithm maintains a function $\tilde{f}$ and a partial assignment $\tilde{\mu}$. Initially, $\tilde{f}=f$ and $\tilde{\mu}$ is empty. The algorithm contains at most $\deg(f)$ rounds: in each but the last round, we extend $\tilde{\mu}$ by fixing at most $md$ additional bits, and reduce the degree of $\tilde{f}$ by at least 1. Specifically, in each round, the algorithm first constructs a maximal set $\Ss$ of disjoint maximum monomials of $\tilde{f}$ (so any maximum monomial of $\tilde{f}$ intersects with at least one monomial in $\Ss$); Then
\begin{enumerate}
    \item If $|\Ss|>m$, then stop and return $\tilde{\mu}$;
    \item Otherwise, fix all variables appearing in $\mathcal{S}$ while keeping the
    new $\tilde{f}$ not identically zero. To do this, process each variable $x_j$ in
    $\mathcal{S}$ one by one. For each, choose a value $b \in \{0,1\}$ such
    that $\tilde{f}$ remains not identically zero after setting $x_j$ as $b$. Such a choice
    always exists since $\tilde{f}$ is not identically zero. Let $\eta$ be the resulting
    partial assignment, and update $\tilde{\mu}$ as $\tilde{\mu} \cdot \eta$ and
    $\tilde{f}$ as $\tilde{f}^\eta$. Here, the function $\tilde{f}^{\eta}(x)$ is defined as $\tilde{f}(x^\eta)$. Now, if $\deg(\tilde{f}) = 0$, then stop and
    return $\tilde{\mu}$.
\end{enumerate}

We now analyze the algorithm. First, we claim that the final $\tilde{\mu}$ satisfies either condition (a) or (b). This is because that: 
\begin{itemize}
    \item If the algorithm stops because
    $\deg(\tilde{f}) = 0$, then $\tilde{f}$ is a constant function that is not zero, say
    $\tilde{f}(x)\equiv c \neq 0$. Therefore, for any $x$, we have
    $f(x^{\tilde{\mu}}) = \tilde{f}(x) = c \neq 0$,
    and condition (a) is satisfied.
    \item If the algorithm stops because $|\mathcal{S}| > m$, then given any $x$, for each maximum
    monomial $\ell$ of $\tilde{f}$ in $\mathcal{S}$, by \Cref{lem:alon}, there exists a partial assignment
    $\mu_\ell$ on the $\leq d$ variables of $\ell$ such that $ f(x^{\mu_\ell \cdot \tilde{\mu}})=\tilde{f}(x^{\mu_\ell})
    \neq 0$. Recalling that the monomials in $\mathcal{S}$ are
    disjoint, we conclude that condition (b) is satisfied. 
\end{itemize}
Next, we show that $|\mu|\leq md^2$, and therefore finish the proof. Since any maximum monomial of $\tilde{f}$ intersects with at least one monomial in $\Ss$, fixing all variables in $\Ss$ reduces $\deg(\tilde{f})$ by at least 1. Hence, the number of round is at most $d$. Moreover, in each round, $|\mu|$ increases by at most $md$. So we have $|\mu|\leq d\cdot (md)=md^2$.
\end{proof}

\section{Impossibility of Perfect-Complete Quantum PKE}



This section will prove that perfect-complete QPKE schemes do not exist
in QROM. More formally, we have the following theorem.
\begin{theorem}[Restate of \Cref{thm:thm1}] \label{thm:qpke}
    For any perfect-complete QPKE in QROM, which makes $d$ queries to a random
    oracle $H:[2^n]\rightarrow\{0,1\}$ during each of $\gen,\enc$ and $\dec$, there exists an adversary Eve that can break the scheme
    w.p.~$1-O(\epsilon)$ by making $O\left( d^7 \log(d/\epsilon) /
    \epsilon^4+nd^2/\epsilon^2\right)$ queries to $H$. 
\end{theorem}

The remainder of this section presents the proof of Theorem \ref{thm:qpke}. It is well-known that: given a perfect-complete QPKE scheme $(\gen^H, \enc^H, \dec^H)$, one can construct a perfect-complete two-round key agreement protocol between two parties, Alice and Bob, as follows.
\begin{enumerate}
    \item Alice computes $(\pk,\sk)\gets \gen^{H}(1^\lambda)$ and sends
    $m_0:=\pk$ to Bob. Denote this stage by $\calA_0$.
    \item Bob randomly chooses $k_B\in\{0,1\}$, computes
    $\ct\gets\enc^{H}(m_0,k_B)$, sends $m_1:=\ct$ to Alice and outputs $k_B$. Denote
    this stage by $\calB$.
    \item Alice computes $k_A\gets \dec^{H}(\sk,m_1)$ and outputs $k_A$. Denote 
    this stage by $\calA_1$.
\end{enumerate}
Each stage of the key agreement makes at most $d$ queries. Thus breaking this QKA also breaks QPKE. 
We now
construct an eavesdropper Eve that sees $(m_0, m_1)$ and guesses the agreed key
w.p.~$1-O(\epsilon)$ by making $O(\poly(n,d,1/\epsilon))$ queries to $H$. Eve's attack algorithm consists of three steps, as summarized in \Cref{fig:algo}. The rest of the section is a detailed analysis of Eve's attack.

\begin{figure}[!t]
    \centering
\begin{mdframed}
\textbf{Algorithm (Eve’s Attack).}
\par\textbf{Input:} random oracle $H:[N]\to\{0,1\}$ with $N=2^n$, transcripts $(m_0,m_1)$, parameters $d,\epsilon$.
\par\textbf{Output:} key $k_E$ such that $\Pr[k_E=k_A]=1-O(\epsilon)$.
\begin{enumerate}
  \item \textbf{Step 1: Identify $\mathcal{B}$’s heavy queries.}
  
  Initialize a classical query record $R_E\gets\emptyset$, and repeats the following
  $\tfrac{d^6}{\epsilon^4}\log\tfrac{d^6}{\epsilon^5}$ times:
  \begin{enumerate}
    \item Sample $t\leftarrow[d]$, simulate $\mathcal{B}^H(m_0)$ up to its $t$-th oracle query,
          measure the query register to obtain $i\in[2^n]$.
    \item Query $i$ to $H$ classically and add $(i,H(i))$ to $R_E$.
  \end{enumerate}
  
  \item \textbf{Step 2: Sample a fake secret key.}
  
  Let $\A$ denote the registers of $\calA_0$, and $\B$ denote the registers of $\calB$ right before $\calB$ performs the final measurement.
  \begin{enumerate}
  \item Run $\mathcal{B}^H(m_0)$ repeatedly for $4dn/\epsilon^2$ times, stopping just before the final measurement, to obtain
  copies $\B_1,\ldots,\B_{4dn/\epsilon^2}$ of $\B$.
   \item Let $j$ be the index\footnote{Existence of such $j$ is guaranteed by \Cref{lemma:permuation-invariance}; see \Cref{sec:step2} for detailed analysis.} such that 
    $I(\A:\B|R_E, \B_1,\ldots,\B_j)\leq \frac{\epsilon^2}{\ln 2}$, and $\E:=(R_E,\B_1,\ldots,\B_j)$. Apply the reconstruction channel
  $\mathcal{T}:\E\!\to\!\E\otimes\A'$ from
  \Cref{thm:quantum_mutual_info_operational}
  to generate a fake copy $\A'$ of $\A$, then measure $\A'$ to obtain a fake secret key $\sk'$.
  \end{enumerate}

  \item \textbf{Step 3: Reprogram the oracle and run $\mathcal{A}_1$.}
  
  Define polynomials $g(H):=\Pr[\mathcal{A}_0^H\!\rightarrow(\sk',m_0)]$ and
  $f(H):=g(H^{R_E})$.
  \begin{enumerate}
    \item If $g$ is the zero polynomial, abort.
    \item Otherwise, apply \Cref{lem:reprogram} to $f$
          with $m=d^2/\epsilon^2$ to obtain a partial assignment
          $\mu:[N]\!\to\!\{0,1,\star\}$ of size at most $d^4/\epsilon^2$. 
    \item 
    Reprogram the oracle as
        \(
        \tilde{H}(i):=H^{\mu}(i)=\begin{cases} 
        \mu(i) & \text{if } i\in\supp(\mu)\\
        H(i) & \text{otherwise}
        \end{cases}.
        \)
     \item Run $k_E\gets\mathcal{A}_1^{\tilde H}(\sk',m_1)$, and output $k_E$.
  \end{enumerate}
\end{enumerate}
\end{mdframed}
 \caption{Summary of Eve's attack algorithm.}
    \label{fig:algo}
\end{figure}

\subsection{\texorpdfstring{Step 1: Identify $\calB$'s heavy queries.}{Step 1: Identify B's heavy queries.}}
The first step is to identify Bob's heavy queries, i.e., inputs with large query
weight. 
These heavy queries will be kept unchanged 
when reprogramming the oracle in later steps, in order to ensure that the reprogrammed oracle
will be indistinguishable from the real oracle from Bob's perspective.

Specifically, in this step, Eve computes a query record
$R_E:=\{(i_E,H(i_E))\}$ by repeating the following process 
$\frac{d^6}{\epsilon^4}\log \frac{d^6}{\epsilon^5}$
times:
\begin{enumerate}
    \item Randomly choose $t\leftarrow[d]$, simulate
    $\calB^H(m_0)$ to its $t$-th query to the oracle, and measure the input
    register, obtain outcome $i\in[2^n]$;
    \item Classically query $i$ to the oracle and add $(i, H(i))$ to $R_E$.
\end{enumerate}

We claim that, with high probability, Eve can identify all of Bob's heavy queries. Formally,
\begin{lemma} \label{lem:wb}
Let $q_i$ be the query weight of input $i$ when running $\calB(m_0)$ on $H$, and 
$W_B:=\{i: q_i\geq \epsilon^4/d^5\}$. Then $\Pr[W_B\not\subseteq R_E]\leq \epsilon$.
\end{lemma}

\begin{proof}
    For each $i\in W_B$, it would be measured w.p.~at least $\epsilon^4/d^6$ at
    each repetition. Thus the probability that it is not measured is bounded by
    \begin{align*}
    \Pr[i\notin R_E]\leq 
    \left(1-\frac{\epsilon^4}{d^6}\right)^{\frac{d^6}{\epsilon^4}\log \frac{d^6}{\epsilon^5}}\leq \epsilon^5/d^6.
    \end{align*}
    Since $\sum_{i}q_i=d$, we have $|W_{B}|\leq (\sum_i q_i)/(\epsilon^4/d^5)= d^6/\epsilon^4$ by Markov's inequality. Thus by a union bound, we have $\Pr[W_B\not\subseteq R_E]\leq \sum_{i\in W_B}\Pr[i\notin R_E]\leq \epsilon$.
\end{proof}

\subsection{Step 2: Sample a fake secret key} \label{sec:step2}
The next step is to obtain a fake secret key that is indistinguishable from the
real secret key from Bob's perspective. We take an information-theoretic
approach: first reduce the mutual information between Alice and Bob's registers
conditioned on Eve's registers to a small value, and then apply the
reconstruction channel in \Cref{thm:quantum_mutual_info_operational} to sample a
fake secret key.

Specifically, consider the time 
right before
$\calB$ performs the final measurement. Let $\A$ denote the registers of
$\calA_0$ and $\B$ denote the registers of $\calB$. Eve repetitively runs
$\calB^H(m_0)$ and stops right before the final measurement for $4dn/\epsilon^2$
times, which yields $4dn/\epsilon^2$ copies
$\B_1,\B_2,\ldots,\B_{4dn/\epsilon^2}$ of $\B$. Observe that $\B, \B_1, \B_2,
\ldots$, $\B_{4dn/\epsilon^2}$ are permutation invariant. By \Cref{lem:entropyB} and
\Cref{lemma:permuation-invariance}, there exists a 
$0\leq j \leq 4dn/\epsilon^2$
such that 
\[
    I(\A:\B|R_E,\B_1,\B_2,\cdots,\B_j)\leq \frac{S(\A)}{4dn/\epsilon^2} \leq \frac{2d(n+1)}{4dn/\epsilon^2} \leq  \frac{\epsilon^2}{\ln 2}.
\]
Note $j$ is determined by the mixed state of $\A\B$ and Eve, averaged over the distribution of all possible oracles, as $H$ is already traced out in \Cref{lemma:permuation-invariance}. Thus $j$ can be computed by computationally unbounded Eve without making oracle queries to $H$\footnote{
$j$ can be explicitly computed by the following (inefficient) procedure:
(i) compute the density matrix $\rho$ of $\A\B\E$, averaged over the uniform distribution of oracles;
(ii) For each $j$,
compute 
$\text{CMI}_j:=I(\A:\B|R_E,\B_1,\B_2,\cdots,\B_j)_\rho$, which is a function of $\rho$;
(iii) output a $j$ such that $\text{CMI}_j
\leq 
\epsilon^2/\ln 2$, of which existence is guaranteed by \Cref{lemma:permuation-invariance}.
One can also guess a uniformly random $j$ if we do not require finding the key with probability close to $1$.}.

Eve only keeps the $j$ copies of $\B$, so that $\E:= (R_E, \B_1, \B_2, \ldots, \B_j)$.
By \Cref{thm:quantum_mutual_info_operational}, there exists a quantum channel $\mathcal{T}:\E\to\E\otimes \A'$ which generates a fake copy $\A'$ of
$\A$ such that 
\[TD(\rho_{ABE},\rho_{A'BE})\leq \sqrt{\ln 2 \cdot I(\A:\B|\E)}< \epsilon,\]
where $\rho_{ABE}$ is the state of system $\A\B\E$ and $\rho_{A'BE}$ is state of
system $\A'\B\E$. Similar to $j$, the channel $\mathcal{T}$ only depends on the mixed state of $\A\E\B$, can thus be implemented by computationally unbounded Eve without making queries to $H$. Since $\A'$ contains a register $\sk'$ storing the secret key,
Eve applies $\mathcal{T}$ on $\E$ to generate $\A'$, and uses $\sk'$ as the fake secret key. 


\begin{lemma} \label{lem:view}
Let $\view_{ABE}:=(\sk, m_0, k_B, m_1, R_E)$, $\view_{A'BE}:=(\sk', m_0, k_B,
m_1$, $R_E)$, $D_{ABE}$ denote the distribution of $\view_{ABE}$, and $D_{A'BE}$
denote the distribution of $\view_{A'BE}$. We have 
\[
    \Pr_{\view_{A'BE}\leftarrow D_{A'BE}}[\view_{A'BE}\notin\supp(D_{ABE})]\leq 2\epsilon.
\]
\end{lemma}

\begin{proof}
    $\view_{ABE}$ and $\view_{A'BE}$ are obtained from performing measurement in
    computational basis on the corresponding registers of state $\rho_{ABE}$ and
    $\rho_{A'BE}$ respectively. Since $TD(\rho_{ABE},\rho_{A'BE})\leq \epsilon$,
    we have that $TV(D_{ABE},D_{A'BE})\leq \epsilon$ by the operational meaning
    of trace distance. By~\Cref{lem:support}, we have\\
    $\Pr_{\view_{A'BE}\leftarrow D_{A'BE}}[\view_{A'BE}\notin\supp(D_{ABE})]\leq
    2\epsilon$.
\end{proof}

The above lemma shows that, with high probability, the tuple $(\sk', m_0, k_B, m_1$, $R_E)$ corresponds to a valid execution and is therefore compatible with some oracle $H'$. Here, we say that a tuple $(\sk', m_0, k_B, m_1, R_E)$ and an oracle $H'$ are compatible if (i) running the QKA protocol on $H'$ generates the view $(\sk', m_0, k_B, m_1)$ with non-zero probability; and (ii) the list of input-output pairs $R_E$ is consistent with $H'$. If Eve had access to such an oracle $H'$, it could just compute $k_E \leftarrow \dec^{H'}(\sk', m_1)$ to break the QKA protocol, since the perfect completeness ensures that $k_E = k_B$.

\subsection{\texorpdfstring{Step 3: Reprogram the oracle and run $\calA_1$}{Step 3: Reprogram the oracle and run A1}}
In the final step, Eve first obtains a reprogrammed oracle $\tilde{H}$ by modifying $O(d^4/\epsilon^2)$ entries of the original oracle $H$ (the details of the reprogramming will be specified later), and then computes $k_E$ by running the decryption algorithm $\mathcal{A}_1$ on $\tilde{H}$ using $\sk'$. In some cases, $\tilde{H}$ is compatible with the tuple $(\sk', m_0, k_B, m_1, R_E)$, in which case we have $k_E = k_B$. However, this compatibility does not always hold\footnote{Unlike \cite{own:ITCS}, our reprogramming step does not always guarantee under $\tilde{H}$, the algorithm outputs $\sk'$ with non-zero probability.}. Nevertheless, we can argue that $(\sk', m_0, k_B, m_1, R_E)$ is compatible with another (possibly unknown) oracle $H'$ that is very close to $\tilde{H}$ and agrees with it on the heavy queries made by $\mathcal{A}_1$. In this case, by \Cref{lem:BBBV}, we can still conclude that $k_E = k_B$ with high probability.

Here are the details of reprogramming. Let $N:=2^n$ and fix $\sk'$ and $m_0$. Let $g(H)$ denote the probability $\Pr[\calA_0^H\rightarrow (\sk',m_0)]$ where
$H\in\{0,1\}^N$ is treated as the truth table of oracle $H$.
Define
$f(H):=g(H^{R_E})$ where we abuse $R_E$ as a partial assignment that assigns the
$i$-th bit of $H$ as $y$ for all $(i, y)\in R_E$. Since $\calA_0$ makes at
most $d$ queries, we have $\deg(f)\leq 2d$ by \Cref{lem:poly_method}. 

If $g$ is a zero polynomial, then Eve aborts\footnote{As we will see, this will never happen unless $\view_{A'BE}\notin\supp(D_{ABE})$.}. Otherwise, Eve applies
\Cref{lem:reprogram} on polynomial $f$ by setting $m=d^2/\epsilon^2$ and obtains
a partial assignment $\mu:[N]\to\{0,1,\star\}$ of size at most $
m\cdot (2d)^2=d^4/\epsilon^2$. We can assume $\supp(\mu) \cap \supp(R_E)=\emptyset$ 
because changing $x_i$ for $i\in\supp(R_E)$ has no effect on the value of
$f(x)=g(x^{R_E})$. Then Eve reprograms the oracle as
\[
\tilde{H}(i):=H^\mu(i)=\begin{cases} 
\mu(i) & \text{if } i\in\supp(\mu)\\
H(i) & \text{otherwise}
\end{cases}.
\]

We have the following lemma. Intuitively, by Lemma~\ref{lem:BBBV}, $\tilde{H}$ is likely to be compatible with Bob’s view $(m_1, k_B)$, as it differs from $H$ on only a small fraction of Bob's query weight.
\begin{lemma}\label{thm:key_compatible}
    For the reprogrammed oracle $\Tilde{H}$ defined above and any quantum
    algorithm $\Bs$ making $d$ queries to the oracle, 
    \begin{align*}
    \Pr_{(k_{B},m_1)\leftarrow \Bs^{H}(m_0)}\left[(k_{B},m_1)\in\supp\left(\Bs^{\tilde{H}}(m_0)\right)\middle\vert \view_{A'BE}\in\supp(D_{ABE}) \right]\geq 1-O(\epsilon),
    \end{align*}
    where we slightly abuse the notation $\Bs^{H}(m_0)$ for the output distribution of the algorithm $\Bs$.
\end{lemma}
\begin{proof} 
    Combining~\Cref{lem:view}
    and~\Cref{lem:wb} we have that
    \begin{align*}
    &\Pr[\view_{A'BE}\in\supp(D_{ABE})\land W_B\subseteq R_{E}]\\
    \geq &    1-\Pr[\view_{A'BE}\notin \supp(D_{ABE})]-\Pr[W_B\not\subseteq R_{E}]\\
    \geq &    1-O(\epsilon).
    \end{align*}
    
    Now consider when $\view_{A'BE}\in \supp(D_{ABE})$ and $W_B\subseteq R_{E}$.
    Since $\view_{A'BE}\in \supp(D_{ABE})$, $(\sk', m_0, R_E)$ is valid under
    some oracle, which implies $f(x)$ is not identically zero. Then Eve will not
    abort and $\tilde{H}$ is well-defined. We have
        \begin{align*}
            TV\left(\Bs^{\tilde{H}}(m_0),\Bs^{H}(m_0)\right) & \leq 4||\ket{\psi_d}-\ket{\phi_d}|| 
            \leq 8\sqrt{d}\sqrt{\sum_{i\colon\Tilde{H}(i)\neq H(i)}q_i} \\
            &\leq 8\sqrt{d}\sqrt{\frac{\epsilon^4}{d^5}|\{i\colon\Tilde{H}(i)\neq H(i)\}|}\\
            &\leq 8\sqrt{d}\sqrt{\frac{\epsilon^4}{d^5}\cdot \frac{d^4}{\epsilon^2}}=O\left(\epsilon\right),
        \end{align*}
        where $\ket{\psi_d}$ and $\ket{\psi_d}$ are the states of
        $\Bs^{\tilde{H}}(m_0)$ and $\Bs^{H}(m_0)$ respectively, and $q_i$ is the
        query weight of input $i$ when running $\Bs$ on $H$. The first
        inequality comes from~\cite[Theorem~3.1]{BBBV97}, the second inequality
        is~\Cref{lem:BBBV}, the third inequality is because $H$ and $\tilde{H}$
        only differ on inputs that are outside $W_B$, and the last inequality is
        because $|\{i\colon\Tilde{H}(i)\neq H(i)\}|\leq |\mu|\leq
        d^4/\epsilon^2$.
        
        By~\Cref{lem:support}, we have that
        \begin{align*}
            \Pr_{(k_{B},m_1)\leftarrow \Bs^{H}(m_0)}\left[(k_{B},m_1)\in\supp\left(\Bs^{\tilde{H}}(m_0)\right)\middle\vert \view_{A'BE}\in\supp(D_{ABE})\land W_B\subseteq R_{E}\right]
            \geq 1-O\left(\epsilon\right).
        \end{align*}
        The final statement follows from a conditional probability formula and Lemma \ref{lem:wb}.
\end{proof}

\subsection{Putting things together}
Now, we are ready to prove Theorem \ref{thm:qpke}.
\begin{proof}[Proof of \Cref{thm:qpke}]



We will prove that by the 3-step attack algorithm described above, Eve will output $k_E$ such
that $\Pr[k_E=k_B]=1-O(\epsilon)$. First of all, by \Cref{lem:view}, we have
that
\begin{equation} \label{eq:first}
    \Pr[\view_{A'BE}\in\supp(D_{ABE})]\geq 1-O(\epsilon).
\end{equation}

Now consider the case when $\view_{A'BE}\in\supp(D_{ABE})$, which implies that
$(\sk', m_0, k_B, m_1, R_E)$ will be a valid execution under some oracle. In this case, 
the function $f$ is not identically zero, and Eve will obtain a partial assignment $\mu$ without aborting. By \Cref{lem:reprogram}, one of the following cases must hold:
\begin{enumerate}
    \item[(a)] $f(H^\mu)\not=0$. 
    As $\tilde{H}=H^\mu$, we have
    $\Pr[(\sk',m_0)\gets\calA_0^{\tilde{H}}]=f(H^\mu)$ is non-zero, which means
    $(\sk',m_0)\in \supp(\calA_0^{\tilde{H}})$.
    \item[(b)] There exist $d^2/\epsilon^2$ pairwise disjoint partial assignments
    $\mu_1,\ldots, \mu_{d^2/\epsilon^2}$ of size $\leq 2d$ such that
    $f(H^{\mu_\ell\cdot\mu})\neq 0$ for all $\ell\in[d^2/\epsilon^2]$. We can assume that for
    all $\ell$, $\supp(\mu_\ell)\cap \supp(R_E)=\emptyset$ as changing $x_i$ for
    $i\in\supp(R_E)$ has no effect on the value of $f(x)=g(x^{R_E})$. Then
    observe that $H^{\mu_\ell\cdot \mu}$ is the truth table of the following oracle
    \[
        \tilde{H}_\ell(i):=H^{\mu_\ell\cdot \mu}(i)=\begin{cases} 
        \mu_\ell(i) & \text{if } i\in\supp(\mu_\ell) \setminus \supp(\mu) \\
        \tilde{H}(i) & \text{otherwise}
        \end{cases}.
    \]
    Thus $\Pr[(\sk',m_0)\gets\calA_0^{\tilde{H}_\ell}]=f(H^{\mu_\ell\cdot \mu})$ is non-zero, which means
    $(\sk',m_0)\in \supp(\calA_0^{\tilde{H}_\ell})$.
\end{enumerate}
Next, we argue that in both cases, Eve will output $k_E=k_B$ with probability
$1-O(\epsilon)$.

\emph{Case (a)} We argue that $\view_{A'BE}=(\sk', m_0, k_B, m_1, R_E)$ is compatible with $\tilde{H}$
with high probability. Obviously, $\tilde{H}$ is consistent with $R_E$. Moreover,
\begin{enumerate}
    \item From perspective of $A'$, $(\sk',m_0)\in \supp(\calA_0^{\tilde{H}})$
    implies that $(\sk',m_0)$ is compatible with $\tilde{H}$. 
    \item From perspective of Bob, $\Bs^{H}(m_0)$ represents a distribution over
    key-message pairs ${(k_B, m_1)}$. Now suppose we run the algorithm
    $\Bs^{\tilde{H}}(m_0)$ instead. According to~\Cref{thm:key_compatible}, with
    probability at least $1-O(\epsilon)$, a pair 
    $(k_B, m_1)$ produced by
    $\Bs^{H}(m_0)$ will also lie within the support of $\Bs^{\tilde{H}}(m_0)$.
\end{enumerate}
So, in particular, $\view_{A'B}$ is a valid execution under $\tilde{H}$ with probability $1-O(\epsilon)$. Conditioned on that $\view_{A'B}$
is valid under $\tilde{H}$, the perfect completeness implies that $k_E=\calA_1^{\tilde{H}}(\sk',m_1)$ must equal $k_B$. Thus we have $\Pr[k_E=k_B]=1-O(\epsilon)$ in Case (a).

\emph{Case (b)} Let $w_i$ be the query weight of input $i$ when running
$\calA_1(\sk',m_1)$ on $\tilde{H}$. Note that $\sum_i w_i\leq d$ because
$\calA_1$ makes at most $d$ queries to $H$. Since
$\mu_1,\ldots,\mu_{d^2/\epsilon^2}$ are disjoint partial assignments, there must
exist a $\ell^*\in[d^2/\epsilon^2]$ such that 
\begin{equation} \label{eq:wa}
\sum_{i\in\supp(\mu_{\ell^*})} w_i\leq \frac{d}{d^2/\epsilon^2}=\frac{\epsilon^2}{d}.
\end{equation}
For simplicity, let $H'$ denote $\tilde{H}_{\ell^*}$. First, imagine that Eve
runs $\mathcal{A}_1(\sk', m_1)$ on $H'$ and obtains a key $k_E'$. Observe that
$H'$ and $H$ differ by at most $|\mu_{l^*}\cdot \mu|\leq
2d+O(d^4/\epsilon^2)=O(d^4/\epsilon^2)$ positions and $\supp(\mu_{l^*}\cdot
\mu)\cap\supp(R_E)=\emptyset$. By the same argument as in \Cref{thm:key_compatible}, 
we have
\begin{equation} \label{eq:kb}
    \Pr_{(k_B,m_1)\gets \calB^{H}(m_0)}\left[(k_{B},m_1)\in\supp\left(\Bs^{H'}(m_0)\right)\middle\vert\view_{A'BE}\in\supp(D_{ABE}) \right]=1-O(\epsilon).
\end{equation}
Then, by the same argument as in Case (a), we have 
\begin{equation} \label{eq:ka}
    \Pr_{k_E'\gets \calA_1^{H'}(\sk',m_1)}\left[k_{E}'=k_B\middle\vert (k_{B},m_1)\in\supp\left(\Bs^{H'}(m_0)\right),\view_{A'BE}\in\supp(D_{ABE}) \right]=1.
\end{equation}
Combining Eqs.~\eqref{eq:first}, \eqref{eq:kb}, and \eqref{eq:ka}, we have
$\Pr_{k_E'\gets \calA_1^{H'}(\sk',m_1)}\left[k_{E}'=k_B\right]=1-O(\epsilon)$,
which means Eve will find the key with probability $1-O(\epsilon)$ if it runs
$\calA_1(\sk', m_1)$ on oracle $H'$. However, Eve knows only the existence of
$H'$, but does not know how to access it. The next step is to argue that by
running $\calA_1(\sk', m_1)$ on oracle $\tilde{H}$ instead, as done in the attack
algorithm, Eve can also find the key with high probability. For any fixed real
oracle $H$, we have that 
\begin{align*}
TV\left(\calA_1^{H'}(\sk',m_1), \calA_1^{\tilde{H}}(\sk',m_1)\right)
&\leq 4 ||\ket{\psi_d}-\ket{\phi_d}|| \\
&\leq 8\sqrt{d}\sqrt{\sum_{i:H'(i)\not=\tilde{H}(i)}w_i} \\
&\leq 8\sqrt{d}\sqrt{\sum_{i\in\mu_{\ell^*}}w_i} 
\leq 8\sqrt{d}\sqrt{\frac{\epsilon^2}{d}}=O(\epsilon)
\end{align*}
where $\ket{\psi_d}$ and $\ket{\psi_d}$ are the states of
$\calA_1^{H'}(\sk',m_1)$ and $\calA_1^{\tilde{H}}(\sk',m_1)$ right before the
final measurement, respectively. The first inequality comes
from~\cite[Theorem~3.1]{BBBV97}, the second inequality is~\Cref{lem:BBBV}, the
third inequality is by definition of $H'$, and the last inequality is
Eq.~\eqref{eq:wa}.

Since the distribution of $k_E\gets \calA_1^{\tilde{H}}(\sk',m_1)$ and
$k_E'\gets\calA_1^{H'}(\sk',m_1)$ are $O(\epsilon)$-close for any $H$, replacing
$k_E'$ with $k_E$ in Eq.~\eqref{eq:ka} will only cause $O(\epsilon)$ loss of the
probability, i.e.,
\begin{equation} \label{eq:kaeps}
    \Pr_{k_E\gets \calA_1^{\tilde{H}}(\sk',m_1)}\left[k_E=k_B\middle\vert (k_{B},m_1)\in\supp\left(\Bs^{H'}(m_0)\right),\view_{A'BE}\in\supp(D_{ABE}) \right]=1-O(\epsilon).
\end{equation}
Combining Eqs.~\eqref{eq:first}, \eqref{eq:kb}, and \eqref{eq:kaeps}, we have
$\Pr_{k_E\gets \calA_1^{\tilde{H}}(\sk',m_1)}[k_E=k_B]=1-O(\epsilon)$. 

Finally, we analyze the query complexity of Eve's attack algorithm. Step 1 requires at most
$(d+1)\cdot d^6\log
(d^6/\epsilon^5)/\epsilon^4=O(d^7\log(d/\epsilon)/\epsilon^4)$ queries. Step 2
requires at most $d\cdot 4dn/\epsilon^2=O\left(d^2n/\epsilon^2\right)$ queries. Step 3
requires $d$ queries. Thus
Eve makes $O\left( d^7 \log(d/\epsilon) / \epsilon^4+nd^2/\epsilon^2\right)$
queries in total.
\end{proof}

\subsection{Extending to quantum public key and ciphertext}

By further reviewing
our proof, we can extend the impossibility result to quantum $m_0$ and $m_1$. 
Specifically, as long
as the public key $\ket{m_0}$ is a pure state that is uniquely determined by
the secret key, and Eve can access polynomially many copies of $\ket{m_0}$,
the attack algorithm still works, with a few minor modifications detailed below.
\begin{itemize}
    \item Steps 1 and 2 of the attack require Eve to run
    $\mathcal{B}^H(\ket{m_0})$ for polynomial times. As Eve can obtain
    polynomially many copies of $\ket{m_0}$, these two steps are doable and the related analysis still holds. 
    \item In Step 3, since Eve does not have the full description of the quantum $m_0$, and thus is not capable of identifying the polynomial that represents the probability of outputting $(\sk', m_0)$.
    However, the key
    observation is that since $\ket{m_0}$ is uniquely determined by the secret
    key, we can instead define the polynomial $f$ as the probability of the oracle
    outputting $\sk'$ alone, i.e., 
    \[g(H):=\Pr[\calA_0^H\rightarrow \sk'],\quad f(H):=g(H^{R_E}).\] 
    Then Eve applies \Cref{lem:reprogram} on this 
    $f$, obtains a
    reprogrammed oracle $\tilde{H}$, and finally outputs
    \(k_E\gets \calA_1^{\tilde{H}}(\sk',\rho_{m_1}).\)
\end{itemize}
The proof is almost the same as before, and we only sketch the main ideas here. 
\begin{enumerate}
    \item From the perspective of $\A'$, the reprogramming the oracle using $f$
    will guarantee $\tilde{H}$ in Case (a) (and $H'$ in Case (b)) to produce
    $\sk'$ with non-zero probability. We argue that the fake secret key $\sk'$
    will produce the real public key $\ket{m_0}$ with high probability: Since
    $\B$ is not affected by channel $\Ts$, we uncompute $\Bs$ on state
    $\rho_{A'B}$. As the uncomputation does not increase trace distance, we can
    see that the state of $(\sk',\ket{m_0})$ is $\epsilon$-close to that of
    $(\sk, \ket{m_0})$. Thus with probability $1-\epsilon$, $\sk'$ will produce
    $\ket{m_0}$.
    \item From the perspective of Bob, the output state of $\Bs^H(\ket{m_0})$
    is a cq-state $\rho_B=\sum_{k_B} p_{k_B}\ket{k_B}\bra{k_B}\otimes
    \rho_{k_B}$ where $\rho_{k_B}$ is the state of $m_1$ conditioned on $k_B$.
    Similarly for $\Bs^{\tilde{H}}(\ket{m_0})$ in Case (a) (and
    $\Bs^{H'}(\ket{m_0})$ in Case (b)), we express the output state as $\sigma_B=\sum_{k_B}
    p_{k_B}'\ket{k_B}\bra{k_B}\otimes \sigma_{k_B}$. By using \Cref{lem:BBBV} as in
    the proof of \Cref{thm:qpke}, 
    we will get $TD(\rho_B, \sigma_B)\leq
    O(\epsilon)$. Let $\rho_{B}':=\sum_{k_B}p_{k_B}'\ket{k_B}\bra{k_B}\otimes \rho_{k_B}$, $D_B:=\{p_{k_B}\}$ and $D_{B'}:=\{p_{k_{B}'}\}$. Then we have
    \begin{align*}
        TD(\rho_B',\sigma_B)\leq& TD(\rho_B,\rho_B') + TD(\rho_B,\sigma_B) 
        = TV\left(D_B,D_{B'}\right) + TD(\rho_B,\sigma_B) \\
        \leq & 2TD(\rho_B,\sigma_B) \leq O(\epsilon),
    \end{align*}
    where the first inequality is the triangle inequality, the second inequality is
    because partial trace does not increase trace distance, and the last
    inequality is $TD(\rho_B, \sigma_B)\leq O(\epsilon)$. Since
    $TD(\rho_B',\sigma_B)=\mathbb{E}[TD(\rho_{k_B},\sigma_{k_B})]$, by Markov
    inequality, we have 
    \[
    \Pr_{k_B\gets D_{B}}[TD(\rho_{k_B}, \sigma_{k_B})\leq \sqrt{\epsilon}] \geq 1-O(\sqrt{\epsilon}).
    \]    
\end{enumerate}
By the above argument, with probability $1 - O(\sqrt{\epsilon})$, 
the tuple $(\sk', \ket{m_0}, k_B, \rho_{k_B}$, $R_E)$ is $O(\sqrt{\epsilon})$-close to a tuple $(\sk', \ket{m_0}, k_B, \sigma_{k_B}, R_E)$ which is compatible with oracle $\tilde{H}$ in Case (a) (and $H'$ is Case (b)).
In Case (a), by perfect completeness, it follows that $k_E = k_B$ with probability $1 - O(\sqrt{\epsilon})$. In Case (b), firstly for $k_E' \gets \As_1^{H'}(\sk', \sigma_{k_B})$, we have $k_E' = k_B$ with probability $1 - O(\sqrt{\epsilon})$. We then apply \Cref{lem:BBBV} to the algorithm $\As_1(\sk', \sigma_{k_B})$ under oracles $\tilde{H}$ and $H'$ to conclude that $k_E = k_B$ with probability $1 - O(\sqrt{\epsilon})$. We remark that although \Cref{lem:BBBV} is stated for pure-state inputs, it also applies to the mixed-state input $\sigma_{k_B}$, since we can always assume the input to be the purification of $\sigma_{k_B}$.


Recall the IND-CPA security notion from~\Cref{def:PKE_qpk}. If the public key is
a pure state, the adversary algorithm $\Es$ can obtain polynomial number of
copies of $\ket{\pk}$. Given any two plaintexts $m_0 \neq m_1$, we can create a
one-bit key agreement by designating the ciphertext as the second message
(choosing between $\rho_{\ct_0}$ and $\rho_{\ct_1}$). Therefore, by executing
our modified attack algorithm, we can break the IND-CPA security with an
advantage of $1-O(\sqrt{\epsilon})$.

\begin{theorem}[Restate of \Cref{thm:thm2}] \label{thm:qpkeqpk} 
For any perfect-complete QPKE with quantum
    public key in QROM, which makes $d$ queries to the random oracle $H:[2^{n}]\to \{0,1\}$ during each of $\skgen,\enc$ and $\dec$, and no queries during $\pkgen$, there
    exists an adversary Eve that can break the scheme with probability $ 1-O (\sqrt{\epsilon})$ by
    making $O\left( d^7 \log(d/\epsilon) / \epsilon^4+nd^2/\epsilon^2\right)$
    queries to $H$. 
\end{theorem}

Theorem \ref{thm:thm3} is immediately implied by Theorem \ref{thm:qpkeqpk}.

\bibliography{ref}
\bibliographystyle{alpha}

\appendix

\section{\texorpdfstring{Missing Proofs in \Cref{sec:helper_lemmas1}}{Missing Proofs}} \label{sec:proofs}

For the convenience of readers, we provide the proofs of two lemmas from \cite{own} that are used in the main
text. 

\subsection{\texorpdfstring{Proof of \Cref{lem:entropyB}}{Proof of Entropy Upperbound}}
    We can realize the quantum query unitary $U_H$ via a quantum communication
    protocol involving two parties: the algorithm $\calA$ and Oracle. To execute $U_H$, the protocol proceeds
    as follows: 
    
    \begin{enumerate} 
        \item $\calA$ sends both its input register and output register,
    $n+1$ qubits in total, to Oracle;
        \item Oracle applies the unitary $U_H$ on these $n+1$ qubits and then returns
    them to $\calA$. 
    \end{enumerate} 
    By the subadditivity of entropy, each such quantum communication can increase the entropy of $\calA$'s whole register
    by at most $2(n+1)$. In addition, applying local unitary $U_i$ does not change the entropy. Since $\calA$'s register $\A$ initially contains a pure state (with zero entropy), it follows that $S(\A)_\rho\leq 2d(n+1)$ after
    $d$ such rounds. 

\subsection{\texorpdfstring{Proof of \Cref{lemma:permuation-invariance}}{Proof of Permutation Invariance}}
The following basic facts below will be used.

\begin{factthm}[\cite{wilde2011classical}] \label{fact:separable_basic}
    If $\rho_{\A\B}$ is a separable state, then $S(\A|\B)\geq 0$.
\end{factthm}

\begin{factthm}[Chain rule]\label{fact:chain_rule}
    $I(\B_1,\B_2,\cdots,\B_t:\A\mid \C)=\sum_{i=1}^t I(\B_i:\A\mid \C,\B_1,\cdots,\B_{i-1})$.
\end{factthm}

\begin{proof}[Proof of \Cref{lemma:permuation-invariance}]
    By the chain rule for conditional mutual information (Fact \ref{fact:chain_rule}), we have
    \begin{equation}\label{eq:non-interactive-eq1}
    \sum_{i=1}^t I(\B_i : \A \mid \C, \B_1, \ldots, \B_{i-1}) = I(\B_1, \ldots, \B_t : \A \mid \C).
    \end{equation}
    Moreover, we have
    \begin{equation}\label{eq:non-interactive-eq2}
    I(\B_1, \ldots, \B_t : \A \mid \C) = S(\A \mid \C) - S(\A \mid \C, \B_1, \ldots, \B_t) \leq S(\A \mid \C) \leq S(\A),
    \end{equation}
    where the inequalities follow from \Cref{fact:separable_basic} and the
    non-negativity of $I(\A:\C)=S(\A)-S(\A\mid \C)$. Combining
    \eqref{eq:non-interactive-eq1} and \eqref{eq:non-interactive-eq2}, it
    follows that there exists some $i\in [t]$ for which
    \[
    I(\B_i : \A \mid \C, \B_1, \ldots, \B_{i-1}) \leq \frac{S(\A)}{t}.
    \]
    Finally, by permutation invariance, we have
    \[
    I(\B_i : \A \mid \C, \B_1, \ldots, \B_{i-1}) = I(\B_t : \A \mid \C, \B_1, \ldots, \B_{i-1}).
    \]
    This completes the proof.
\end{proof}

\end{document}